\documentclass[10pt,conference]{IEEEtran}
\usepackage{amsmath,amsfonts,amsthm}
\usepackage{algorithmic}
\usepackage{algorithm}
\usepackage{array}
\usepackage{textcomp}
\usepackage{stfloats}
\usepackage{url}
\usepackage{color}
\usepackage{verbatim}
\usepackage{graphicx}
\usepackage{cite}
\usepackage{tikz}

\newtheorem{lem}{Lemma}
\newtheorem{prop}{Proposition}

\newcommand{\minus}{\scalebox{0.5}[1.0]{$-$}}
\IEEEoverridecommandlockouts
\begin{document}

\title{Performance Characterization of Continuous Reconfigurable Intelligent Surfaces

\author{

\IEEEauthorblockN{Amy S. Inwood\IEEEauthorrefmark{1}, Peter J. Smith\IEEEauthorrefmark{2}, Mahmoud Alaaeldin\IEEEauthorrefmark{1},~and Michail Matthaiou\IEEEauthorrefmark{1}}

\IEEEauthorblockA{\IEEEauthorrefmark{1}Centre for Wireless Innovation (CWI), Queen’s University Belfast, Belfast BT3 9DT, U.K.}

\IEEEauthorblockA{\IEEEauthorrefmark{2}School of Mathematics and Statistics, Victoria University of Wellington, Wellington, New Zealand}

\IEEEauthorblockA {Emails: \{a.inwood, m.alaaeldin, m.matthaiou\}@qub.ac.uk, peter.smith@vuw.ac.nz}
    \vspace{-0.45cm}
}    

\thanks{This work is a contribution by Project REASON, a UK Government funded project under the Future Open Networks Research Challenge (FONRC) sponsored by the Department of Science Innovation and Technology (DSIT). It was also supported by the U.K. Engineering and Physical Sciences Research Council (EPSRC) (grants No. EP/X04047X/1 and EP/X040569/1). The work of M. Alaaeldin and M. Matthaiou was supported by the European Research Council (ERC) under the European Union’s Horizon 2020 research and innovation programme (grant agreement No. 101001331).}}

\maketitle

\begin{abstract}
We consider a reconfigurable intelligent surface (RIS) that can implement a phase rotation continuously over the whole surface rather than via a finite number of discrete elements. Such an RIS can be considered a design for future systems where advances in metamaterials make such an implementation feasible or as the limiting case where the number of elements in a traditional RIS increases in a given area. We derive the optimal RIS design for the single-user (SU) scenario assuming a line-of-sight (LoS) from the RIS to the base station (BS) and correlated Rayleigh fading for the other links. We also derive the associated optimal signal-to-noise ratio (SNR) and its mean, a bound on the mean spectral efficiency (SE), an approximation to the SNR outage probability and an approximation to the coefficient of variation for the investigation of channel hardening.
\end{abstract}


\section{Introduction}
Interest in RISs continues to grow, driven by the ability of an RIS to manipulate the wireless channel with very little power consumption. Typical RIS implementations involve an array of large numbers of reflective elements, where each element is tunable and adjusts the phase of the reflected signal. An RIS can be used to enhance the SNR and other system performance metrics, avoid blockages and increase the channel rank. 

Despite the focus of most work on discrete elements, the fundamental ideas of RISs relate to controlled reflections at a surface and many emerging works consider densely populated arrays or continuous surfaces. Examples of these include holographic massive multiple-input and multiple-output (MIMO) \cite{demir_channel_2022, sanguinetti_wavenumber_2023} holographic RIS \cite{huang_holographic_2020, wan_terahertz_2021}, large intelligent surfaces (LIS) \cite{dardari_communicating_2020, williams_communication_2020} and stacked intelligent metasurfaces (SIM) \cite{an_stacked_2024}. In \cite{smith_continuous_2024}, a statistical analysis was performed on a system where a single antenna user equipment (UE) communicates with a finite dimension continuous surface receiver. This system consisted of a single link. In terms of a typical three-link RIS system with UE-BS, UE-RIS and RIS-BS channels, the work in the discrete domain in \cite{singh_optimal_2022} has derived an optimal RIS design for the fundamental SU scenario and evaluated the mean SNR, a bound on the mean SE and an approximation to the outage probability. However, none of this analysis has been conducted in the continuous domain. Therefore, an open research problem is the statistical analysis of key performance metrics when an RIS surface is truly continuous (i.e. each point on the surface can effect the desired phase shift). Note that this work is non-trivial since it requires continuous spatial transformations.
Hence, in this work, we make the following contributions:
\begin{itemize}
    \item  We derive the optimal SU continuous RIS design. 
    \item Using this optimal design:
    \begin{itemize}
        \item We derive the associated optimal SNR and its mean, and the approximate variance.
        \item We bound the mean SE.
        \item We approximate the outage probability.
        \item We approximate the coefficient of variation (CV) to investigate the effects of channel hardening.
    \end{itemize}
\end{itemize}

The benefits of this analysis are twofold: These new expressions provide insights into continuous RIS performance while also providing upper bounds on the performance of a discrete RIS as the number of elements in a fixed area tends to infinity. This work also sets up the system from a continuous perspective for further work, such as investigations into more complex channel models, different UE configurations, and the impact of surface imperfections.

\textit{Notation}: Upper and lower boldface letters represent matrices and vectors, respectively, and $\mathbf{M}_{r,s}$ is the $(r,s)$-th element of $\mathbf{M}$. Two exceptions to this convention occur in the notation for the RIS channels as discussed in Section \ref{sec:sysmodel}. 
The statistical expectation is denoted $\mathbb{E}[\cdot]$, $\mathrm{Var}[\cdot]$ denotes the variance and $\mathbb{P}(A)$ denotes the probability of event $A$; $\mathcal{CN}(\boldsymbol\mu,\mathbf{Q})$ represents a complex Gaussian distribution with mean $\boldsymbol\mu$ and covariance matrix $\mathbf{Q}$; ${}_2F_1(a,b,c;z)$ is the Gaussian hypergeometric function, $\gamma(\cdot,\cdot)$ is the incomplete gamma function, $\Gamma(\cdot)$ is the gamma function while $E(\cdot)$ and $K(\cdot)$ are the elliptic integrals of the first and second kind; $(\cdot)^*$ and $(\cdot)^\dagger$ represent the complex conjugate and Hermitian transpose operators, respectively. The angle of a complex number, $z$, is denoted $\angle{z}$.

\section{System Model}\label{sec:sysmodel}
We consider the uplink RIS-aided system in Fig. \ref{fig:LCR_Paper_System_Diagram}, where a continuous rectangular RIS surface, of width $W$ and height $H$, is located near a BS with $M$ antennas. One single-antenna UE is located near to both.
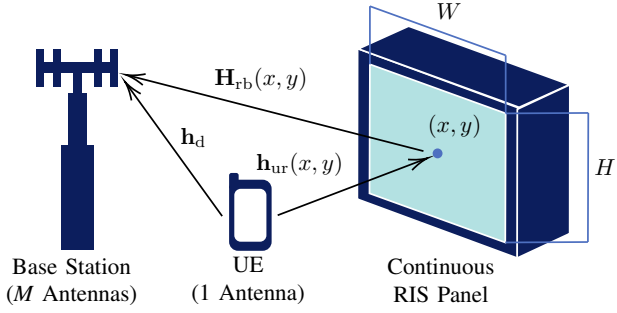
\begin{figure}[ht]
    \centering
    \resizebox{0.48\textwidth}{!}{
        \tikzset{every picture/.style={line width=0.75pt}} 

\begin{tikzpicture}[x=0.75pt,y=0.75pt,yscale=-1,xscale=1]

\draw  [color={rgb, 255:red, 255; green, 255; blue, 255 }  ,draw opacity=1 ][fill={rgb, 255:red, 12; green, 30; blue, 93 }  ,fill opacity=1 ] (213.56,28.88) -- (305.88,62.97) -- (305.73,149.77) -- (213.41,115.68) -- cycle ;
\draw  [color={rgb, 255:red, 255; green, 255; blue, 255 }  ,draw opacity=1 ][fill={rgb, 255:red, 179; green, 225; blue, 226 }  ,fill opacity=1 ] (220.64,37.95) -- (298.88,66.26) -- (298.65,140.7) -- (220.41,112.39) -- cycle ;
\draw  [color={rgb, 255:red, 255; green, 255; blue, 255 }  ,draw opacity=1 ][fill={rgb, 255:red, 12; green, 30; blue, 93 }  ,fill opacity=1 ] (242.96,15.95) -- (335.62,49.67) -- (306.64,62.76) -- (213.98,29.03) -- cycle ;
\draw  [color={rgb, 255:red, 255; green, 255; blue, 255 }  ,draw opacity=1 ][fill={rgb, 255:red, 12; green, 30; blue, 93 }  ,fill opacity=1 ] (305.59,63.88) -- (335.86,49.74) -- (335.76,135.57) -- (305.49,149.7) -- cycle ;
\draw  [draw opacity=0][fill={rgb, 255:red, 12; green, 30; blue, 93 }  ,fill opacity=1 ] (42.5,85) -- (61.5,85) -- (61.5,145) -- (42.5,145) -- cycle ;
\draw  [draw opacity=0][fill={rgb, 255:red, 12; green, 30; blue, 93 }  ,fill opacity=1 ] (47,54.5) -- (57.5,54.5) -- (57.5,104) -- (47,104) -- cycle ;
\draw  [draw opacity=0][fill={rgb, 255:red, 12; green, 30; blue, 93 }  ,fill opacity=1 ] (49.5,37.5) -- (54.5,37.5) -- (54.5,87) -- (49.5,87) -- cycle ;
\draw  [draw opacity=0][fill={rgb, 255:red, 12; green, 30; blue, 93 }  ,fill opacity=1 ] (28.5,37) -- (75.5,37) -- (75.5,43) -- (28.5,43) -- cycle ;
\draw  [draw opacity=0][fill={rgb, 255:red, 12; green, 30; blue, 93 }  ,fill opacity=1 ] (28.5,29.5) -- (33,29.5) -- (33,51) -- (28.5,51) -- cycle ;
\draw  [draw opacity=0][fill={rgb, 255:red, 12; green, 30; blue, 93 }  ,fill opacity=1 ] (38,29.5) -- (42.5,29.5) -- (42.5,51) -- (38,51) -- cycle ;
\draw  [draw opacity=0][fill={rgb, 255:red, 12; green, 30; blue, 93 }  ,fill opacity=1 ] (61.5,29.5) -- (66,29.5) -- (66,51) -- (61.5,51) -- cycle ;
\draw  [draw opacity=0][fill={rgb, 255:red, 12; green, 30; blue, 93 }  ,fill opacity=1 ] (71,29.5) -- (75.5,29.5) -- (75.5,51) -- (71,51) -- cycle ;
\draw  [color={rgb, 255:red, 12; green, 30; blue, 93 }  ,draw opacity=1 ][fill={rgb, 255:red, 12; green, 30; blue, 93 }  ,fill opacity=1 ] (140,109.7) .. controls (140,107.1) and (142.1,105) .. (144.7,105) -- (158.8,105) .. controls (161.4,105) and (163.5,107.1) .. (163.5,109.7) -- (163.5,140.3) .. controls (163.5,142.9) and (161.4,145) .. (158.8,145) -- (144.7,145) .. controls (142.1,145) and (140,142.9) .. (140,140.3) -- cycle ;
\draw  [draw opacity=0][fill={rgb, 255:red, 255; green, 255; blue, 255 }  ,fill opacity=1 ] (142.5,113.45) .. controls (142.5,111.41) and (144.16,109.75) .. (146.2,109.75) -- (157.3,109.75) .. controls (159.34,109.75) and (161,111.41) .. (161,113.45) -- (161,136.55) .. controls (161,138.59) and (159.34,140.25) .. (157.3,140.25) -- (146.2,140.25) .. controls (144.16,140.25) and (142.5,138.59) .. (142.5,136.55) -- cycle ;
\draw  [color={rgb, 255:red, 12; green, 30; blue, 93 }  ,draw opacity=1 ][fill={rgb, 255:red, 12; green, 30; blue, 93 }  ,fill opacity=1 ] (140,102) .. controls (140,101.5) and (140.4,101.1) .. (140.9,101.1) -- (143.6,101.1) .. controls (144.1,101.1) and (144.5,101.5) .. (144.5,102) -- (144.5,108.7) .. controls (144.5,109.2) and (144.1,109.6) .. (143.6,109.6) -- (140.9,109.6) .. controls (140.4,109.6) and (140,109.2) .. (140,108.7) -- cycle ;
\draw    (167.5,125) -- (252.13,92.71) ;
\draw [shift={(254,92)}, rotate = 159.12] [color={rgb, 255:red, 0; green, 0; blue, 0 }  ][line width=0.75]    (10.93,-3.29) .. controls (6.95,-1.4) and (3.31,-0.3) .. (0,0) .. controls (3.31,0.3) and (6.95,1.4) .. (10.93,3.29)   ;
\draw    (251.5,88) -- (81.94,44.5) ;
\draw [shift={(80,44)}, rotate = 14.39] [color={rgb, 255:red, 0; green, 0; blue, 0 }  ][line width=0.75]    (10.93,-3.29) .. controls (6.95,-1.4) and (3.31,-0.3) .. (0,0) .. controls (3.31,0.3) and (6.95,1.4) .. (10.93,3.29)   ;
\draw    (135.5,125.5) -- (80.67,49.62) ;
\draw [shift={(79.5,48)}, rotate = 54.15] [color={rgb, 255:red, 0; green, 0; blue, 0 }  ][line width=0.75]    (10.93,-3.29) .. controls (6.95,-1.4) and (3.31,-0.3) .. (0,0) .. controls (3.31,0.3) and (6.95,1.4) .. (10.93,3.29)   ;
\draw  [draw opacity=0][fill={rgb, 255:red, 78; green, 112; blue, 192 }  ,fill opacity=1 ] (256.8,89.32) .. controls (256.8,87.76) and (258.08,86.48) .. (259.64,86.48) .. controls (261.21,86.48) and (262.48,87.76) .. (262.48,89.32) .. controls (262.48,90.89) and (261.21,92.17) .. (259.64,92.17) .. controls (258.08,92.17) and (256.8,90.89) .. (256.8,89.32) -- cycle ;
\draw [color={rgb, 255:red, 78; green, 112; blue, 192 }  ,draw opacity=1 ]   (298.88,66.26) -- (346.5,66) ;
\draw [color={rgb, 255:red, 78; green, 112; blue, 192 }  ,draw opacity=1 ]   (298.65,140.7) -- (346.27,140.44) ;
\draw [color={rgb, 255:red, 78; green, 112; blue, 192 }  ,draw opacity=1 ]   (346.5,66) -- (346.27,140.44) ;
\draw [color={rgb, 255:red, 78; green, 112; blue, 192 }  ,draw opacity=1 ]   (298.88,32.76) -- (298.88,66.26) ;
\draw [color={rgb, 255:red, 78; green, 112; blue, 192 }  ,draw opacity=1 ]   (220.64,4.45) -- (298.88,32.76) ;
\draw [color={rgb, 255:red, 78; green, 112; blue, 192 }  ,draw opacity=1 ]   (220.64,4.45) -- (220.64,37.95) ;

\draw (5,148) node [anchor=north west][inner sep=0.75pt]   [align=left] {\begin{minipage}[lt]{63.39pt}\setlength\topsep{0pt}
\begin{center}
Base Station\\(\textit{M} Antennas)
\end{center}

\end{minipage}};
\draw (112,147.5) node [anchor=north west][inner sep=0.75pt]   [align=left] {\begin{minipage}[lt]{55.46pt}\setlength\topsep{0pt}
\begin{center}
UE\\(1 Antenna)
\end{center}

\end{minipage}};
\draw (222,148.5) node [anchor=north west][inner sep=0.75pt]   [align=left] {\begin{minipage}[lt]{57.16pt}\setlength\topsep{0pt}
\begin{center}
Continuous \\RIS Panel
\end{center}

\end{minipage}};
\draw (110,72.9) node [anchor=north west][inner sep=0.75pt]    {$\mathbf{h}_{\mathrm{d}}$};
\draw (129.5,38.4) node [anchor=north west][inner sep=0.75pt]    {$\mathbf{H}\mathrm{_{rb}}\mathnormal{( x,y)}$};
\draw (252.5,65.4) node [anchor=north west][inner sep=0.75pt]    {$\mathnormal{( x,y)}$};
\draw (257.5,1.4) node [anchor=north west][inner sep=0.75pt]    {$\mathnormal{W}$};
\draw (348.5,90.9) node [anchor=north west][inner sep=0.75pt]    {$H$};
\draw (153.5,87.4) node [anchor=north west][inner sep=0.75pt]    {$\mathbf{h}\mathrm{_{ur}}\mathnormal{( x,y)}$};

\end{tikzpicture}
    }
    \caption{System model showing the uplink channels.}
    \label{fig:LCR_Paper_System_Diagram}
\end{figure}

Let $\mathbf{h}_{\mathrm{d}} \in \mathbb{C}^{M\times 1}$ be the channel vector from the UE to BS. We assume spatially correlated Rayleigh fading so that $\mathbf{h}_{\mathrm{d}} \sim \mathcal{CN}(\mathbf{0},\beta_{\mathrm{d}}\mathbf{R}_{\mathrm{d}})$, where $\beta_{\mathrm{d}}$ is the channel gain and $\mathbf{R}_{\mathrm{d}}$ is the correlation matrix. For the RIS links, the channel is defined in terms of an arbitrary point, $(x,y)$, on the rectangular surface, such that $0 \le x \le W$ and $0 \le y \le H$. The RIS-BS link is assumed to be LoS and is defined by the vector $\mathbf{H}_{\mathrm{rb}}(x,y) = \sqrt{\beta_{\mathrm{rb}}}\mathbf{a}_\mathrm{b}\mathbf{a}_\mathrm{r}^*(x,y)$, where  $\mathbf{a}_\mathrm{b}$ is the steering vector for the LoS ray at the BS and $\beta_{\mathrm{rb}}$ is the channel gain. While the BS steering vector is an $M \times 1$ vector due to the $M$ BS antennas, the use of a continuous RIS makes the phase shift at the surface a function of position, so the steering function $\mathbf{a}_\mathrm{r}(x,y) \in \mathbb{C}$ is required. Note that the LoS assumption is a common RIS scenario as motivated in \cite{nadeem_asymptotic_2020}.

For the UE-RIS channel, we assume spatially correlated Rayleigh fading, so that at two points on the surface, $\mathbf{h}_{\mathrm{ur}}(x,y)\in \mathbb{C}$ and $\mathbf{h}_{\mathrm{ur}}(x',y')\in \mathbb{C}$ are both $\mathcal{CN}(0,\beta_{\mathrm{ur}})$ with correlation $\rho(x,y,x',y')=\mathbb{E}[\mathbf{h}_{\mathrm{ur}}(x,y)\mathbf{h}^*_{\mathrm{ur}}(x',y')]/\beta_{\mathrm{ur}}$, where $\beta_{\mathrm{ur}}$ is the channel gain. Finally, at the RIS, let $\mathbf{\Phi}(x,y)$ be the reflection coefficient at the point $(x,y)$. 

The received signal at the BS is therefore
\begin{equation}
    \mathbf{r}{=}\!\!\left(\!\mathbf{h}_{\mathrm{d}}\!+\!\!\!\int\limits_{x=0}^W\int\limits_{y=0}^H\mathbf{H}_{\mathrm{rb}}(x,y)\mathbf{\Phi}(x,y) \mathbf{h}_{\mathrm{ur}}(x,y)dydx\!\right)\!s\!+\!\mathbf{n}, \label{eq:channel}
\end{equation}
where $\mathbf{n} \sim \mathcal{CN}(\mathbf{0},\sigma^2\mathbf{I}_M)$ is additive white Gaussian noise while $E_s=\mathbb{E}[|s|^2]$ is the transmitted energy of symbol $s$.

\subsection{Optimal RIS design}
The RIS coefficients can be expressed as
\begin{equation}
	\label{eq:Phi}
	\mathbf{\Phi}(x,y) = \mathbf{a}_\mathrm{r}(x,y)e^{j\boldsymbol{\psi}(x,y)}e^{-j\angle \mathbf{h}_{\mathrm{ur}}(x,y)},
\end{equation}
where $\boldsymbol{\psi}(x,y)$ is an arbitrary phase angle. Substituting \eqref{eq:Phi} and the definition of $\mathbf{H}_{\mathrm{rb}}$ into  \eqref{eq:channel} and simplifying gives 
\begin{equation}
    \mathbf{r} {=}\!\!\left(\!\mathbf{h}_{\mathrm{d}}{+}\sqrt{\beta_{\mathrm{rb}}}\mathbf{a}_\mathrm{b}\!\!\int\limits_{x=0}^W\int\limits_{y=0}^H|\mathbf{h}_{\mathrm{ur}}(x,y)|e^{j\boldsymbol{\psi}(x,y)} dydx\!\right)\!s\!+\!\mathbf{n}. \label{eq:channel2}
\end{equation}
Defining the double integral in \eqref{eq:channel2} as $ce^{j\nu}$, where $c>0$, gives the simplified result:
\begin{equation}
    \mathbf{r} = \left(\mathbf{h}_{\mathrm{d}}+\sqrt{\beta_{\mathrm{rb}}}\mathbf{a}_\mathrm{b}ce^{j\nu}\right)s + \mathbf{n} =  \mathbf{h}s+\mathbf{n}, \label{eq:channel3}
\end{equation}
from which the SNR-optimal RIS design can be derived in the following Proposition.
\begin{prop}\label{Thm1}
    The optimal RIS design  is given by
    \begin{equation}
	\label{eq:OptPhi}
	\mathbf{\Phi}(x,y) = \omega \mathbf{a}_\mathrm{r}(x,y)e^{-j\angle \mathbf{h}_{\mathrm{ur}}(x,y)},
 \end{equation}
 where $\omega=\tfrac{\mathbf{a}_\mathrm{b}^\dagger\mathbf{h}_{\mathrm{d}}}{|\mathbf{a}_\mathrm{b}^\dagger\mathbf{h}_{\mathrm{d}}|}$ and the corresponding maximum SNR is
 \begin{equation}\label{eq:OptSNR}
 \mathrm{SNR}=\frac{E_s}{\sigma^2}\left(\mathbf{h}_\mathrm{d}^H \mathbf{h}_\mathrm{d}+M\beta_\mathrm{rb}Y^2+2\sqrt{\beta_{\mathrm{rb}}}Y|\mathbf{a}_\mathrm{b}^\dagger\mathbf{h}_{\mathrm{d}}|\right),
  \end{equation}
  where $Y=\int_{x=0}^W\int_{y=0}^H|\mathbf{h}_{\mathrm{ur}}(x,y)| dydx$.
\end{prop}
\begin{proof}
    See Appendix A for the derivation of \eqref{eq:OptPhi} and \eqref{eq:OptSNR}.
\end{proof}

\section{Analysis}
\subsection{Optimal SNR: Exact Mean and Approximate Variance}\label{sec:meanandvar}
From \eqref{eq:OptSNR} and using the results $\mathbb{E}[\mathbf{h}_\mathrm{d}^\dagger\mathbf{h}_\mathrm{d}]=M\beta_\mathrm{d}$ and $\mathbb{E}[|\mathbf{a}_\mathrm{b}^\dagger\mathbf{h}_\mathrm{d}|]=\tfrac{1}{2}\sqrt{\pi\beta_\mathrm{d}\mathbf{a}_\mathrm{b}\mathbf{R}_\mathrm{d}\mathbf{a}_\mathrm{b}}$ from \cite{singh_optimal_2022}, $\mu_1=\mathbb{E}[\mathrm{SNR}]$ is
    \begin{align}
	\label{eq:MeanSNR1}
	\mu_1&\!=\!\frac{E_s}{\sigma^2}\!\left(\!\mathbb{E}\!\left[\mathbf{h}_\mathrm{d}^\dagger\mathbf{h}_\mathrm{d}\right]\!+\!M\beta_\mathrm{rb}\mathbb{E}\!\left[Y^2\right]\!+\!2\sqrt{\beta_\mathrm{rb}}\mathbb{E}[Y]\mathbb{E}\!\left[|\mathbf{a}_\mathrm{b}^\dagger\mathbf{h}_\mathrm{d}|\right]\right)\!,\notag\\
    &\!=\!\frac{E_s}{\sigma^2}\Big( \!M\beta_\mathrm{d}+M\beta_\mathrm{rb}m_2+m_1\sqrt{\pi\beta_\mathrm{rb}\beta_\mathrm{d} \mathbf{a}_\mathrm{b}^H\mathbf{R}_\mathrm{d}\mathbf{a}_\mathrm{b}} \Big),
 \end{align}
where $m_k=\mathbb{E}[Y^k]$. 
\begin{lem}
As $|\mathbf{h}_{\mathrm{ur}}(x,y)|$ is a correlated Rayleigh process over $0 \le x \le W$, $0 \le y \le H$, we have that
    \begin{align}
  &\mathbb{E}[Y]=\tfrac{1}{2}\sqrt{\pi \beta_\mathrm{ur}}HW,\label{EY}\\
  &\mathbb{E}[Y^2]=\!\!\!\!\int\limits_{x'=0}^W \int\limits_{y'=0}^H\int\limits_{x=0}^W\int\limits_{y=0}^H\frac{\pi\beta_\mathrm{ur}}{4} \notag \\ &\qquad\quad\times{}_2F_1\!\!\left(\!\frac{-1}{2},\!\frac{-1}{2};\!1;\!|\rho(x,y,x'\!,y')|^2\!\!\right)\!dydxdy'dx'.\label{EY2}
 \end{align}
\end{lem}
\begin{proof}
    See Appendix B for the derivation of \eqref{EY} and \eqref{EY2}.
\end{proof}
 The integral in \eqref{EY2} is cumbersome to compute, but in practice most correlation models admit considerable simplification.
 \begin{lem}
     When correlation is isotropic, such that $\rho(x,y,x',y')$ is solely a function of the spatial separation, $r=\sqrt{(x-x')^2\!+\!(y-y')^2}$,
    \begin{align}
        & \mathbb{E}[Y^2] = 4\Biggl\{ \int_{0}^{H}\!\!rg(r)\bigg[WH\frac{\pi}{2}-(W\!\!+\!\!H)r\!+\!\frac{r^2}{2}\bigg] dr \notag \\
        & + \!\!\int_{H}^{W}\!\!\!\!rg(r)\bigg[\!W\!H\!\sin^{\!-1}\!\!\bigg(\!\frac{H}{r}\!\bigg)\!+\!W\!\sqrt{r^2\!\!-\!\!H^2}\!-Wr-\!\frac{H^2}{2}\bigg] dr \notag \\
        & -\!\!\int_{W}^{\sqrt{W^2+H^2}}rg(r)\bigg[WH\bigg(\!\cos^{-1}\!\bigg(\!\frac{W}{r}\bigg)\!-\sin^{-1}\!\bigg(\!\frac{H}{r}\bigg)\!\bigg)\! \notag \\
        &+ \frac{W^2+H^2+r^2}{2} - W \sqrt{r^2\!\!-\!\!H^2}- H \sqrt{r^2\!\!-\!\!W^2}\bigg]\!dr\!   \Biggl\}\!,\!\! \label{eq:EY2simple}
    \end{align}
    where $g(r)=\frac{\pi\beta_\mathrm{ur}}{4} {}_2F_1\left(-\frac{1}{2},-\frac{1}{2};1;|\rho(x,y,x',y')|^2\right).$
\end{lem}
\begin{proof}
    See Appendix C for the derivation of (\ref{eq:EY2simple}).
\end{proof}
Hence, for isotropic correlation models, the mean and variance of $Y$ can be obtained from the closed-form result in \eqref{EY} and the use of \eqref{eq:EY2simple} which requires only single numerical integrals. Note that (\ref{eq:EY2simple}) assumes that $H\leq W$, but the formula is general such that $W$ and $H$ can be switched if $W<H$ in the isotropic case.
\subsection{A Simple Bound on the Mean SE}\label{sec:bound}
Jensen's inequality gives the simple bound on the mean SE,
\begin{align}
  \mathbb{E}[\textrm{SE}]&=  \mathbb{E}[\textrm{log}_2(1+\textrm{SNR})], \label{eq:SE} \\ 
 & \le \textrm{log}_2\left(1+\mu_1\right). \label{eq:SEbound}
\end{align}

\subsection{Optimal SNR: Approximate Outage Results}\label{sec:outage}

We believe an exact expression for the SNR distribution is intractable. Therefore, we propose a gamma approximation to the SNR distribution using an approximation technique from \cite{singh_optimal_2022}. To fit the gamma distribution using the method of moments, the mean and variance are required. The mean is given in \eqref{eq:MeanSNR1}, so it suffices to compute the second moment, $\mu_2=\mathbb{E}[\textrm{SNR}^2]$. Squaring \eqref{eq:OptSNR}, we obtain
\begin{align}
  &\mu_2\!=\!\frac{E_s^2}{\sigma^4}\mathbb{E}\Big[(\mathbf{h}_\mathrm{d}^H \mathbf{h}_\mathrm{d})^2\!\!+\!2M\beta_\mathrm{rb}Y^2\mathbf{h}_\mathrm{d}^H \mathbf{h}_\mathrm{d}\!+\!4\sqrt{\!\beta_\mathrm{rb}}Y\mathbf{h}_\mathrm{d}^H \mathbf{h}_\mathrm{d}|\mathbf{a}_\mathrm{b}^\dagger\notag\\
  & \!\times\!\mathbf{h}_{\mathrm{d}}|\!+\!\!M^2\beta_\mathrm{rb}^2Y^4\!\!+\!4M\beta_\mathrm{rb}^{\frac{3}{2}}Y^3 |\mathbf{a}_\mathrm{b}^\dagger\mathbf{h}_{\mathrm{d}}|\!+\!4\beta_\mathrm{rb}Y^2|\mathbf{a}_\mathrm{b}^\dagger\mathbf{h}_{\mathrm{d}}|^2\Big]\!.\!\label{mu2a}
\end{align}
The expectation of the variables in \eqref{mu2a} involving $\mathbf{h}_{\mathrm{d}}$ are found in \cite{singh_optimal_2021}. Using these results and the independence of $Y$ and $\mathbf{h}_{\mathrm{d}}$ gives
\begin{align}
  &\mu_2=\frac{E_s^2}{\sigma^4}\bigg[\beta_\mathrm{d}^2(\mathrm{tr}(\mathbf{R}_\mathrm{d}^2)\!+\!\mathrm{tr}(\mathbf{R}_\mathrm{d})^2)\!+\!2M^2\beta_\mathrm{d}\beta_\mathrm{rb}m_2\!+\!m_1\notag\\
  &\times\!\sqrt{\pi\beta_\mathrm{rb} \beta_\mathrm{d}^3 \mathbf{a}_\mathrm{b}^H\mathbf{R}_\mathrm{d}\mathbf{a}_\mathrm{b}}\left(2M\!+\!\frac{\mathbf{a}_\mathrm{b}^H\mathbf{R}_\mathrm{d}^2\mathbf{a}_\mathrm{b}}{\mathbf{a}_\mathrm{b}^H\mathbf{R}_\mathrm{d}\mathbf{a}_\mathrm{b}}\right)\!+\!M^2\beta_\mathrm{rb}^2m_4\notag\\
    &+2Mm_3\sqrt{\pi\beta_\mathrm{d}\beta_\mathrm{rb}^3\mathbf{a}_\mathrm{b}^H\mathbf{R}_\mathrm{d}\mathbf{a}_\mathrm{b}}+4\beta_\mathrm{d}\beta_\mathrm{rb}m_2\mathbf{a}_\mathrm{b}^H\mathbf{R}_\mathrm{d}\mathbf{a}_\mathrm{b}\bigg].\label{mu2}
  \end{align}
From \eqref{mu2}, only $m_k$ for $k=1,2,3,4$ are required to compute $\mu_2$. The first two moments, $m_1$ and $m_2$, are given in \eqref{EY} and \eqref{EY2} (or \eqref{eq:EY2simple} for the isotropic case). Hence, $m_3$ and $m_4$ are required. While expressions for these moments can be given in terms of multiple integrals, closed-form expressions for higher-order correlations of Rayleigh variables are not available. Hence, we apply a further approximation, modelling $Y$ as a gamma variable. With this approach, we can express $m_3$ and $m_4$ in terms of $m_1$ and $m_2$. Recursively defining the third and fourth moments leads to
\begin{align}
    m_3&=\left(2m_2-m_1^2\right)\frac{m_2}{m_1},\label{m3}\\
    m_4&=\left(3m_2-2m_1^2\right)\left(2m_2-m_1^2\right)\frac{m_2}{m_1^2}.\label{m4}
\end{align}
Using \eqref{m3} and \eqref{m4} we have all the terms in $\mu_1$ and $\mu_2$ and the resulting gamma approximation is $\mathrm{SNR} \sim \mathcal{G}(\alpha_g,\beta_g)$, where $\alpha_g={\mu_1^2}/(\mu_2-\mu_1^2)$, $\beta_g={\mu_1}/(\mu_2-\mu_1^2)$. Hence, the outage approximation is
\begin{equation}
    \mathbb{P}(\mathrm{SNR} \le x) \approx \frac{\gamma(\alpha_g,\beta_g x)}{\Gamma(\alpha_g)}.
\end{equation}

\subsection{Channel Hardening: Approximate Coefficient of Variation}
Channel hardening is a phenomenon where a channel behaves almost deterministically due to spatial diversity \cite{bjornson_massive_2017}. The CV is a known metric to quantify the channel hardening effect, and can be approximated using the gamma fit proposed in Section \ref{sec:outage} so that
\begin{equation} \label{cv_metric}
    \mathrm{CV}^2  = \frac{\mathrm{Var}[||\mathbf{h}||^2]}{\mathbb{E}[\mathrm{||\mathbf{h}||^2}]^2} \approx \frac{\mu_2-\mu_1^2}{\mu_1^2}.
\end{equation}

\section{Numerical Results}\label{sec:numresults}
This section verifies the analytical results and explores the system behaviour. In Figs. \ref{fig:meanSNR}-\ref{fig:channelhardening}, we consider the mean SNR, the mean SE, the outage probability and the channel hardening, respectively. Unless stated otherwise, $M=32$ and the carrier frequency is $5.8$ GHz.  For simulation purposes, we consider two channel correlation models. We use a scaled version of Rayleigh fading sinc correlation,
$\rho_{m,n} = \mathrm{sinc} \left(2\, \kappa \, d_{m,n} \right)$,
and a scaled version of classical Jakes' correlation,
$\rho_{m,n}=J_0(2\pi\kappa\,d_{m,n})$, 
where $\rho_{m,n}$ is the correlation between the antennas/points $m$ and $n$, $d_{m,n}$ is the Euclidean distance between the antennas/points $m$ and $n$ and $\kappa$ is an arbitrary scaling parameter for the investigation of varying degrees of correlation. If not stated differently, $\kappa=1$. The channel gain values were selected according to the distance-based path loss model in \cite{wu_intelligent_2019}, where $\beta =  C_0\left({d}/{D_0}\right)^{-\alpha}$, $D_0$ is the reference distance of 1 m, $C_0$ is the path loss at $D_0$ (-30 dB), $d$ is the link distance in metres and $\alpha$ is the path loss exponent. Assuming an indoor transmission environment with an NLoS direct link, path loss exponents were selected to be $\alpha_{\mathrm{d}} = 6$, $\alpha_{\mathrm{rb}} = 1.7$ and $\alpha_{\mathrm{ur}} = 1.7$ \cite{rappaport_wireless_2001}. Setting $d_\mathrm{rb}$, $d_x$ and $d_y$ (the distance between the BS and RIS, and BS and UE in the $x$ and $y$ directions respectively) in \cite[Fig. 2]{wu_intelligent_2019} gives the link distances.\footnote{Note that in \cite{wu_intelligent_2019}, $d_\mathrm{rb},d_x$ and $d_y$ are denoted $d_0,d$ and $d_v$, respectively.} Unless otherwise specified, $d_\mathrm{rb}=5$ m, $d_x = 30$ and $d_y = 1$ m. 

The steering vector $\mathbf{a}_\mathrm{b}$ corresponds to the vertical uniform rectangular array (VURA) model in  \cite{miller_analytical_2019}. Elements are arranged at intervals of $d_b$ wavelengths at the BS. $M_x$ and $M_z$ are the number of BS elements per row and column respectively, such that $M = M_xM_z$. In these results, we assume $M_x=8$ and $M_z = 4$. $\theta_\mathrm{A}$ and $\phi_\mathrm{A}$ are the corresponding elevation and azimuth angles of arrival (AoAs) at the BS. We assume the RIS is on a $\frac{\pi}{4}$ rad angle with respect to the BS, so $\phi_\mathrm{A} = \frac{\pi}{4}$ rad. We also assume both are at the same height, so $\theta_\mathrm{A}=\frac{\pi}{2}$ rad. For all simulations, $10^5$ replicates were generated.

\subsection{Mean SNR}
\label{sec:meanSNRres}

Figure \ref{fig:meanSNR} verifies (\ref{eq:MeanSNR1}) and investigates the mean SNR performance for a range of RIS surface areas using both the sinc (denoted S) and Jakes' (denoted J) correlation models.

\begin{figure}[ht]
    \centering
    \includegraphics[scale=0.75,trim={0.2cm 0.15cm 0.2 0.2cm},clip]{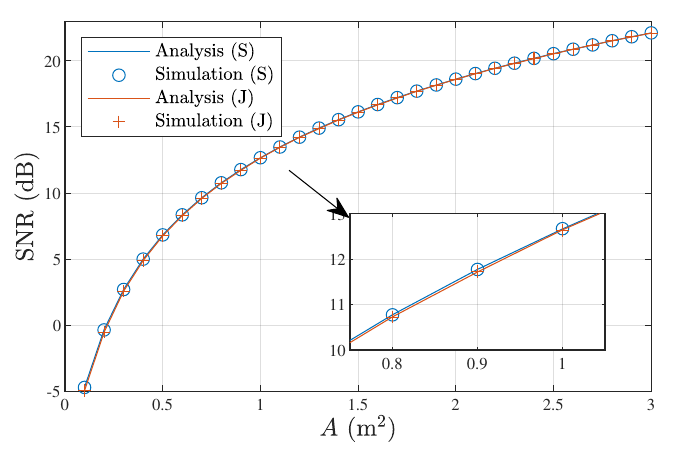}
    \caption{Comparison of the simulated mean SNR and the analytical mean SNR in (\ref{eq:MeanSNR1}) for both the sinc and Jakes' correlation models and a varied RIS area.}
    \label{fig:meanSNR}
\end{figure}

Figure \ref{fig:meanSNR} shows that the mean SNR increases as the RIS surface area increases. A larger RIS surface area can collect more signal energy, improving the SNR. The linear SNR grows approximately quadratically with the area. As shown in (\ref{eq:EY2gf}), $m_2$ in (\ref{eq:OptSNR}) is scaled by $A^2$, and this term dominates the total SNR as the surface area increases. Also, the sinc and Jakes' correlation models result in very similar mean SNR performance.

\subsection{Mean Spectral Efficiency}
Figure \ref{fig:meanSE} investigates the performance of the spectral efficiency bound (SEB) presented in Section \ref{sec:bound}. The Jakes' correlation model is used, and the SE bound is compared to the simulated SE performance for different $\kappa$ values and RIS areas. The arbitrary correlation parameter $\kappa$ is varied from $0$ (giving perfect correlation) to $1$ (giving the typical Jakes' correlation). When $\kappa < 1$, the initial rate of decay and frequency of oscillation of the zero-th order Bessel function is reduced. This increases the correlation between nearby pairs of points on the RIS and allows us to investigate the effects of increased correlation on the SE. Results are plotted for surface areas of $A = W\times H = 0.1, 0.2, 0.3$ and $0.4$ m$^2$, where $W=H=\sqrt{A}$.

\begin{figure}[ht]
    \centering
    \includegraphics[scale=0.75,trim={0.2cm 0.15cm 0.2 0.2cm},clip]{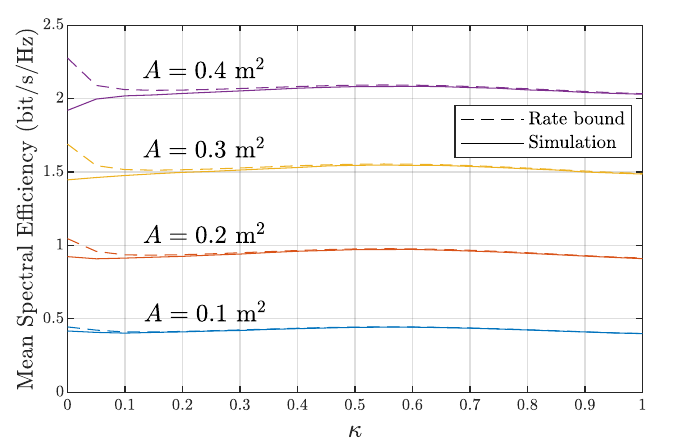}
    \caption{Comparison of the simulated rate and the rate bound provided by Jensen's inequality in Section \ref{sec:bound} for a range of correlations and RIS areas.}
    \label{fig:meanSE}
\end{figure}

While generally the overall correlation decreases as $\kappa$ grows, the SE itself exhibits minor fluctuations. This fluctuation is due to artefacts in the zero-th order Bessel function comprising the Jakes correlation model, which leads to the correlation between points on the surface not decreasing monotonically with distance. 

The SEB accurately approximates the true SE when $\kappa > 0.1$. At $\kappa = 0$, the error is $6 - 16\%$, but at $\kappa = 1$, the error is less than $0.5\%$. The rate bound is most accurate for larger $\kappa$ values, as they result in lower correlations. The correlation models proposed for RIS systems are inherently low, so the SEB will perform well in these scenarios. Improved bound performance at lower correlations agrees with existing literature \cite{zhang_power_2014, inwood_phase_2023}. Applying a Taylor series to (\ref{eq:SE}) leads to
\begin{equation}
    \mathbb{E}[\mathrm{SE}]\!=\!\frac{1}{\mathrm{ln}(2)}\!\bigg[\mathrm{ln}(1+\mathbb{E}[\mathrm{SNR}]) + \frac{\mathrm{Var[SNR]}}{2(1+\mathbb{E}[\mathrm{SNR}])^2} + ...\bigg],
\end{equation}
where the first term is the SEB. The dominant error term (DET) is the largest term not included in the bound. Therefore,
\begin{equation}
    \mathrm{DET} = \frac{\mathbb{E}\!\left[\mathrm{SNR}^2\right]-\mathbb{E}\left[\mathrm{SNR}\right]^2}{2\mathrm{ln}(2)(1+\mathbb{E}[\mathrm{SNR}])^2} = \frac{\mu_2-\mu_1^2}{2\mathrm{ln}(2)(1+\mu_1)^2}. \label{eq:DET}
\end{equation}
Note that when $\mu_1>\!>1$, the term $1$ in the denominator of (\ref{eq:DET}) can be ignored, so that the DET is approximately proportional to CV$^2$, showing that channel hardening improves the tightness of the bound.  Table \ref{tab:SEBerror} compares the DET and the SEB.
\begin{table}[ht]
    \centering
    \caption{Comparison of the SEB and DET for the $A=0.4 \,\mathrm{m}^2$ scenario.}
    \label{tab:SEBerror}
    \begin{tabular}{|c|c|c|c|}
    \hline
     $\kappa$& SEB & DET & DET/SEB (\%) \\
     \hline
     $0$ & 2.2760 &  0.5070 & 22.28\%\\
     $0.1$ & 2.0606 & 0.0497 & 2.41\%\\
     $0.5$ & 2.0908 & 0.0099 & 0.47\%\\
     $1$ & 2.0329 &  0.0048 & 0.24\%\\
     \hline
    \end{tabular}
\end{table}

Note that the DET monotonically decreases as correlation decreases, further verifying the behaviour observed in Fig. \ref{fig:meanSE}.



\subsection{SNR Outage Probability}
Figure \ref{fig:outageprob} verifies the analytical results in Section \ref{sec:outage} and investigates the SNR outage probability. The impacts of both the Jakes' and sinc correlation models, the RIS area and the RIS dimensions are considered. Results are plotted for two different continuous RIS surface areas, $A = W\times H = 0.3\, \mathrm{m}^2$ and $0.4\,\mathrm{m}^2$. For each area and correlation model, results are plotted for a ratio of width to height of $20:1$ and $1:1$.
\begin{figure}[ht]
    \centering
    \includegraphics[scale=0.75,trim={0.2cm 0.15cm 0.2 0.2cm},clip]{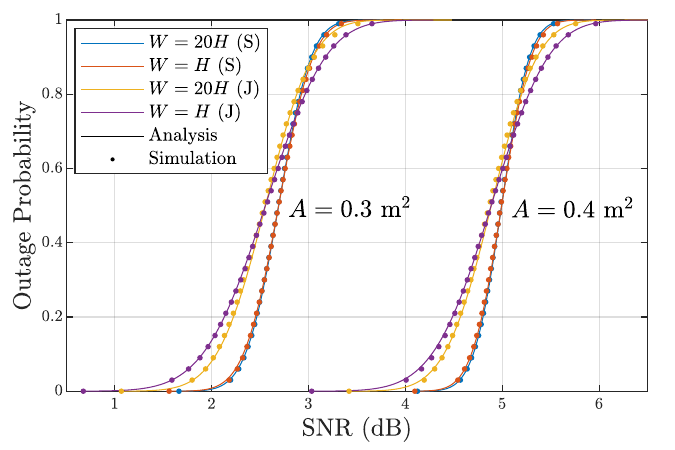}
    \caption{Comparison of the SNR outage probability for varied RIS areas, dimensions and correlation models.}
    \label{fig:outageprob}
\end{figure}

As seen in Section \ref{sec:meanSNRres}, as the surface area of the continuous RIS increases, the SNR increases. A larger RIS can capture and reflect more energy, thus increasing the SNR at which outages are more likely to occur. 

The SNR outage performance avails of lower correlation, which is visible through changes to the dimensions of the RIS for a fixed area. When the width is twenty times larger than the height, the larger gaps between points on the surface at opposite extremities lead to lower correlations when compared to a square RIS. This leads to a more pronounced averaging of the channel variations and a narrower CDF curve, indicating more stable and consistent outage probabilities. 

Similarly, the sinc correlation model results in narrower curves and a higher median SNR. The sinc function decreases at a faster rate with distance than the zeroth order Bessel function comprising the Jakes' model, leading to lower correlation. 

The naturally higher correlations provided by the Jakes' model result in the effects of the RIS dimensions becoming more pronounced. There is a larger gap between the $W=20H$ and $W=H$ Jakes' model curves for both RIS areas than between the corresponding sinc model curves. The increased variability in the channels from less averaging means that varying the RIS dimensions has a larger impact on the outage probability overall. Thus, RIS dimensions are an important consideration, particularly for higher correlation systems.

Finally, the impacts of all metrics discussed are more noticeable when the direct link is weak (i.e $\beta_\mathrm{d}$ is small). When this occurs, the RIS impacts a higher proportion of the total channel, and thus parameter changes can be more clearly seen.

\subsection{Channel Hardening}
Figure \ref{fig:channelhardening} investigates the channel hardening effect of the continuous RIS for various $\kappa$ values, system layouts and RIS surface areas. The sinc and Jakes' correlation models result in similar behaviour, so only the sinc model is shown. The channel hardening measure CV$^2$ is plotted for three layouts: $\{d_y, d_\mathrm{rb}, d_x\} = \{1, 40, 27\}$ (Setup A), $\{1, 40, 53\}$ (Setup B), and $\{1, 5, 27\}$ (Setup C). Note that an increase in channel hardening results in a decrease in CV$^2$.
\label{sec:channelhardening}
\begin{figure}[ht]
    \centering
    \includegraphics[scale=0.75, trim={0.05cm 0.15cm 0.2 0.2cm},clip]{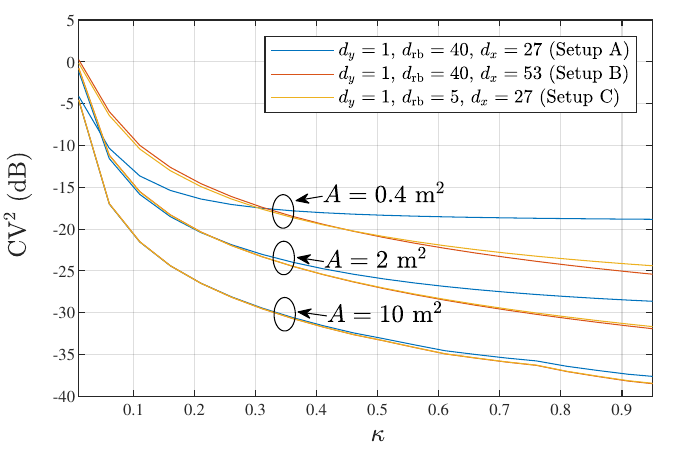}
    \caption{Comparison of the channel hardening capability (measured by CV$^2$) for various system layouts, RIS surface areas and correlations.}
    \label{fig:channelhardening}
\end{figure}

More channel hardening occurs as $\kappa$ increases. The channel hardening concept relies on summing a large number of \textit{independent} and identically distributed random variables (RVs), and an increase in $\kappa$ results in a decrease in the spatial correlation between pairs of points on the RIS. 

The channel hardening effect increases with the RIS surface area. Increasing the surface area leads to the averaging of channel fluctuations over more points, which is analogous to increasing the number of antennas in a massive MIMO system.

While both setups have the same $d_\mathrm{ur}$ and $d_\mathrm{rb}$, Setup A results in less channel hardening than Setup B. The direct channel has a larger effect to the SNR in Setup A due to the lower BS-user distance. The direct channel only has 32 antennas, so it cannot support hardening properties. Therefore, increasing its effect on the overall SNR reduces the channel hardening overall. Additionally, Setup C leads to more channel hardening than Setup A at high $\kappa$ values, despite their equivalent $d_\mathrm{d}$. The shorter $d_\mathrm{rb}$ of Setup C leads to an increased RIS contribution to the SNR and hence, more channel hardening. The difference between the setups decreases as the RIS surface area increases, as the RIS channel dominates the channel hardening.

\section{Conclusion}
This work investigated key statistics and behaviours of a continuous RIS system. We conceived the SNR-optimal reflection coefficient design for a continuous RIS, and expressions for the optimal mean SNR, an upper bound on the mean SE, the approximate SNR outage probability and the approximate CV were derived. Numerical results were simulated to verify these expressions and provide insights into the effects of key parameters on continuous RIS performance, including RIS area, RIS dimensions and correlation. These results also provided an upper bound to the performance of a discrete RIS as the number of elements in a fixed area tends to infinity. 

\section*{Appendix A \\ Proof of Proposition 1}
From \eqref{eq:channel3}, the SNR is given by 
\begin{equation}\label{A1}
\textrm{SNR}=||\mathbf{h}_{\mathrm{d}}+\sqrt{\beta_{\mathrm{rb}}}\mathbf{a}_\mathrm{b}ce^{j\nu}||^2 \frac{E_s}{\sigma^2}.
\end{equation}
Maximizing an expression of the form in \eqref{A1} requires maximizing $c$ and setting $e^{j\nu}=\mathbf{a}_\mathrm{b}^\dagger\mathbf{h}_{\mathrm{d}}/|\mathbf{a}_\mathrm{b}^\dagger\mathbf{h}_{\mathrm{d}}|$. Now,
\begin{align}
    &ce^{j\nu}=\int_{x=0}^W\int_{y=0}^H|\mathbf{h}_{\mathrm{ur}}(x,y)|e^{j\angle\boldsymbol{\psi}(x,y)} dydx\notag\\
    &\implies c\le \int_{x=0}^W\int_{y=0}^H|\mathbf{h}_{\mathrm{ur}}(x,y)| dydx = Y.\label{UB}
\end{align}
Note that the upper bound in \eqref{UB} is obtained when $\boldsymbol{\psi}(x,y)$ is any constant so setting $e^{j\boldsymbol{\psi}(x,y)}=\mathbf{a}_\mathrm{b}^\dagger\mathbf{h}_{\mathrm{d}}/|\mathbf{a}_\mathrm{b}^\dagger\mathbf{h}_{\mathrm{d}}|$ maximizes $c$, yielding $c=Y$, and results in the optimal phase shift \cite{singh_optimal_2022}. Substituting this optimal choice for $\boldsymbol{\psi}(x,y)$ into \eqref{eq:Phi} gives the 

\noindent optimal phase design in \eqref{eq:OptPhi}, and substituting $ce^{j\nu}$ with $Y\mathbf{a}_\mathrm{b}^\dagger\mathbf{h}_{\mathrm{d}}/|\mathbf{a}_\mathrm{b}^\dagger\mathbf{h}_{\mathrm{d}}|$ in \eqref{A1} and expanding the quadratic gives the optimal SNR in \eqref{eq:OptSNR} as required.

\section*{Appendix B \\ Proof of Lemma 1}
Let $r_1 = |\mathbf{h}_{\mathrm{ur}}(x,y)|$ and $r_2 = |\mathbf{h}_{\mathrm{ur}}(x',y')|$. Therefore,
    \begin{equation}
        \mathbb{E}[Y]= \mathbb{E}\!\left[\,\,\int\limits_{x=0}^W\int\limits_{y=0}^H r_1 dydx\right]= \int\limits_{x=0}^W\int\limits_{y=0}^H\mathbb{E}[r_1]dydx. \label{eq:EYint} 
    \end{equation}
    Using the result for the mean amplitude of a correlated Rayleigh process in (4.2) of \cite{miller_complex_1974} allows us to write
    \begin{align}
        \mathbb{E}[Y]&= \int\limits_{x=0}^W\int\limits_{y=0}^H\Gamma\left(\frac{3}{2}\right)\mathrm{Var}[\mathbf{h}_\mathrm{ur}(x,y)]^{1/2}dydx, \notag \\ &=\frac{\sqrt{\pi\beta_\mathrm{ur}}}{2}\int_{x=0}^W1dx\int_{y=0}^H1dy = \frac{\sqrt{\pi\beta_\mathrm{ur}}}{2}HW,\notag
    \end{align}
    as required. Similarly, $\mathbb{E}[Y^2]$ can be written as
    \begin{align}
         &\mathbb{E}[Y^2]= \mathbb{E}\!\left[\,\,\int\limits_{x'=0}^W \int\limits_{y'=0}^H\int\limits_{x=0}^W\int\limits_{y=0}^Hr_1r_2\,dydxdy'dx'\right], \notag \\ & \qquad\,\,\, = \int\limits_{x'=0}^W \int\limits_{y'=0}^H\int\limits_{x=0}^W\int\limits_{y=0}^H \mathbb{E}[r_1r_2] dydxdy'dx'.\label{eq:EY2int}
    \end{align}
Using (4.2) of \cite{miller_complex_1974} to evaluate $\mathbb{E}[r_1r_2]$ leads to
\begin{align}
    &\mathbb{E}[Y^2]\!=\!\!\!\!\!\int\limits_{x'=0}^W \int\limits_{y'=0}^H\int\limits_{x=0}^W\int\limits_{y=0}^H\!\!\frac{\pi}{4}\mathrm{Var}[|\mathbf{h}_{\mathrm{ur}}(x,y)|]^{\tfrac{1}{2}} \mathrm{Var}[|\mathbf{h}_{\mathrm{ur}}(x',y')|]^{\tfrac{1}{2}}\notag\\ &\qquad\qquad\times {}_2F_1\!\!\left(\!\frac{-1}{2},\!\frac{-1}{2};\!1;\!|\rho(x,y,x'\!,y')|^2\!\!\right)dydxdy'dx',\notag \\ &=\!\!\!\int\limits_{x'=0}^W\!\!\!\!\ldots\!\!\!\!\int\limits_{y=0}^H\frac{\pi\beta_\mathrm{ur}}{4}{}_2F_1\!\!\left(\!\frac{-1}{2},\!\frac{-1}{2};\!1;\!|\rho(x,y,x'\!,y')|^2\!\!\right)dydxdy'dx'. \notag
\end{align}
as required.

\section*{Appendix C \\ Proof of Lemma 2}
Using the definition of $g(r)$ for isotropic correlation, 
    \begin{align}
     \mathbb{E}[Y^2]&=\int\limits_{x'=0}^W \int\limits_{y'=0}^H\int\limits_{x=0}^W\int\limits_{y=0}^Hg(r)dydxdy'dx', \notag \\
     &= W^2H^2\!\!\!\int\limits_{x'=0}^W \int\limits_{y'=0}^H\int\limits_{x=0}^W\int\limits_{y=0}^H\frac{g(r)}{W^2H^2}dydxdy'dx',\notag \\
     &= W^2H^2\,\mathbb{E}[g(r)], \notag
 \end{align}
where the expectation is taken over $\frac{1}{W^2H^2}$, i.e., the joint density of two points uniformly located on the surface. Hence,
\begin{equation}
    \mathbb{E}[Y^2] = W^2H^2\int_{r=0}^{\sqrt{W^2+H^2}}\!\!\!g(r)f_s(r)dr,
\label{eq:EY2gf}
\end{equation}
where $f_s(r)$ is the probability density function (PDF) of the separation between two random points in a rectangle of dimensions $W\times H$. From \cite{mathai_random_1998}, 
\begin{equation}
\label{eq:fsr}
    f_s(r)\!=\!\frac{4r}{W^2\!H^2}\!\!
    \begin{cases}
        \frac{\pi WH}{2}{-}(W{+}H)r{+} \frac{r^2}{2},& r<H \vspace{0.5em} \\ \vspace{0.5em}
        \begin{aligned}
            WH\sin^{\!-1}(\tfrac{H}{r})-Wr \\ + W\!\sqrt{r^2{-}H^2}{-}\tfrac{H^2}{2},
        \end{aligned} & H {\leq}\, r\,{<} W \\
        \begin{aligned}
        W\!H\!\left(\!\sin^{\!\!\minus1}\!\!\left(\!\tfrac{H}{r}\!\right)\!{-}\!\cos^{\!\!\minus1}\!\!\left(\!\tfrac{W}{r}\!\right)\!\right) \\{-}\tfrac{W^2+H^2}{2} {+}W\!\sqrt{\!r^2{-}H^2} \\  {-}\tfrac{r^2}{2}{+} H\sqrt{r^2{-}W^2}.
        \end{aligned} & \begin{aligned}
            W {\leq}\,r\,{\leq} \\ \sqrt{\!W^2{+}H^2}
        \end{aligned} 
    \end{cases}
\end{equation}
Substituting (\ref{eq:fsr}) into (\ref{eq:EY2gf}) gives (\ref{eq:EY2simple}) as required.

\bibliographystyle{IEEEtran}
\bibliography{IEEEabrv, referencesCTS.bib}

\end{document}